\documentclass{amsart}
\usepackage[pdftex]{graphicx}
\usepackage{amsthm,amscd,amsmath,amssymb,amsfonts,enumerate}

\usepackage{hyperref}
\usepackage{float}
\usepackage{adjustbox}

\newtheorem{theorem}{Theorem}[section]

\newtheorem{prop}[theorem]{Proposition}
\newtheorem{cor}[theorem]{Corollary}
\newtheorem{fact}[theorem]{Fact}

\newtheorem*{schur}{Schur's Lemma}
\newtheorem*{thm3.4}{Theorem 3.4 of \cite{BarceloEtAl}}
\newtheorem*{thm3.8}{Theorem 3.8 of \cite{BarceloEtAl}}
\newtheorem*{lem3.6}{Lemma 3.6 of \cite{BarceloEtAl}}
\newtheorem*{thm4.2}{Theorem 4.2 of \cite{BarceloEtAl}}

\theoremstyle{definition}
\newtheorem{definition}[theorem]{Definition}
\newtheorem{example}[theorem]{Example}

\theoremstyle{remark}
\newtheorem{remark}[theorem]{Remark}

\newtheorem{notation}[theorem]{Notation}

\numberwithin{equation}{section}

\newcommand{\ind}{\mathord\uparrow} 
\newcommand\rank[4]{\begin{array}{c}#1\\#2\\#3\\#4\end{array}}
\newcommand{\p}{{\bf p}}

\begin{document}

\title[Structured Voting for Committees]{Structured Voting for Structured Committees}

\author[Crisman]{Karl-Dieter Crisman}
\address{Department of Mathematics and Computer Science, Gordon College}
\curraddr{}
\email{karl.crisman@gordon.edu}
\thanks{This work was partly supported by a Gordon College sabbatical leave.}

%\begin{highlights}
%\item Voting on committees structured by subgroups of candidates is very amenable to algebraic analysis
%\item Ballots with a single structured committee inevitably lead to structure in voting systems based on number of agreements with the chosen committee
%\item Ranked ballots of \emph{all} committees allow a huge variety of neutral voting systems
%\item Giving ballots with different internal structures different sets of weights is possible and may be appropriate sometimes
%\end{highlights}

%\textbf{JEL Classification Codes}: C02; D71

%\textbf{MSC Primary 91B12; Secondary 91B14,20C05}

\date{}


\begin{abstract}
There has been much recent work on multiwinner voting systems.  However, sometimes a committee is highly structured, and if we want to vote for such a committee, our voting method should be more structured as well.  

We consider committees consisting of representatives for disjoint parts of a collective; for instance, $n$ departments of an organization might each have $m$ possible representatives to choose from.  However, in our model the \emph{entire population} still gets to vote on the committee, and we need to respect that symmetry, which we do by building on earlier work of Barcelo et al.~using representation theory of wreath products.

The main thrust of this paper is to advertise and catalog the surprising variety of possible points-based (linear) voting methods in this framework.  For instance, if a voter ranks all possible structured committees, it may be reasonable to use weights which depend on the relations between the ranked committees, and not just their order as in `traditional' voting frameworks!
\end{abstract}

\maketitle

\section{Introduction}

We begin by justifying our interest in studying \emph{structured} committees, giving some basic examples, and outlining the rest of the paper.

\subsection{Motivation and Context}

\subsubsection*{Why Committees?}
There has been a recent explosion of interest in multiwinner elections, a relatively `new challenge' for the social choice community (\cite{FaliszewskiEtAlNewChallenge}).  
Nearly every voting system used in practice has the \emph{possibility} of more than one winner\footnote{Possibly made resolute by coin toss.%, see % get ref
}, but in this area the focus is on systems taking voter (or artificial-intelligence) preferences to come up with a selected set of winners of a predetermined size.  Examples might range from voting for municipal committees\footnote{In my own city, one may vote for up to $n$ of the (at most) $2n$ candidates for $n$ seats in the general election, where $n=4$ for councilors-at-large and $n=6$ for the school committee.} to choosing corporate boards to a streaming company selecting five `best' options to pick from for your family movie night.
 
\subsubsection*{Why Structured Committees?}
In a series of intriguing papers starting with \cite{RatliffPublicChoice}, Ratliff describes elections involving \emph{structured} committees.  For a presidential search committee at his institution, each of three academic divisions were to have one of two possible representatives selected; however, the faculty as a whole was allowed to vote (separately) on the representatives for each division.  As it turned out, these preferences were \emph{not} separable; despite essentially universal desire for a mixed-gender committee (particularly since his institution had only recently even begun admitting male students!), the committee selected was all men.

Such structure is inherent in any representative body, but what is unusual here is that the \emph{entire} electorate votes on the \emph{entire} committee, yet with this additional structure.  Such nonseparable structure in a group decision need not be about actual representation, of course.  In an entertainment example, a group of friends going to a large festival might want to attend one sporting event, one play, and one musical performance -- but one could easily imagine the friend who either wants all three outdoors or all three indoors, and would strongly dislike any mixing of location.

For a related paradox, see \cite{AliceBob}.  One important point is that this is not the same kind of committee diversity raised in \cite{ElkindEtAlMultiwinner}
where the objective may either be to have some sort of proportional representation (as with diverse electorates) or simply at least one option for each category regardless of popularity (as with a buffet).  In neither of these cases is the structure necessarily fixed ahead of time, as in our model.

\subsubsection*{Why Structured Voting?}
There are many possible approaches one could take in such circumstances. 
In \cite{RatliffPublicChoice}, Ratliff describes the experience of asking all faculty to rank \emph{all eight} possible committees, and then using the Borda count on submitted preferences.  In \cite{RatliffSelectingDiverse} and \cite{RatliffSaariComplexities} (the latter joint work with Saari), voting methods are analyzed which involve separate point values for a `diversity' candidate, and requiring diverse ballots (where `diverse' here may simply mean non-homogeneous).  

The point with any such system is to join the structure of representative systems with a universal `electorate' in order to achieve an acknowledged structural goal.  One goal of research in this area should be to contribute to the structural understanding of structured voting procedures.

\subsection{Background and Examples}
A big step in this direction was taken with \cite{BarceloEtAl}, building on thesis work in \cite{LeeThesis}.   Suppose there are $n$ departments, each potentially represented by one of $m$ candidates, yielding $m^n$ possible committees.  The authors suppose (as in \cite{RatliffPublicChoice}) that voters will submit a ranking of \emph{all} of these structured committees, and then a series of weights (not necessarily Borda weights) are applied to give a score for each committee, with the highest score yielding a winning committee.  It is important to note that while this could be useful for `diverse committee' voting as in \cite{RatliffPublicChoice}, that additional layer is \emph{not} part of their model (or ours).

To see how more analysis might help, consider the smallest possible case of $m=n=2$, which occurs several times in \cite{BarceloEtAl} and \cite{LeeThesis}.  Here, there are only four possible committees, and voters are asked to rank all four of them.  The committee notation $(1_1,2_1)$ designates selecting the first option in both departments $1$ and $2$.

\begin{example}\label{ex:2wr2first}
Suppose we have three voters, three with preference $(1_1,2_1)\succ (1_2,2_2)\succ (1_1,2_2)\succ (1_2,2_1)$ and two with preference $(1_2,2_1)\succ (1_2,2_2)\succ (1_1,2_2)\succ (1_1,2_1)$.  
If we apply Borda weights $(3,2,1,0)$ to these voters, we see easily that $(1_2,2_2)$ receives ten points and is the winner.
\end{example}

So far, this example isn't really any different from `normal' voting theory.  But consider the \emph{full} set of scores.

\begin{example}\label{ex:2wr2second}
Continuing Example \ref{ex:2wr2first}, we obtain the scoring vector $[9,5,6,10]^T$ in lexicographical order by committee ($(1_1,2_1)$ through $(1_2,2_2)$).  However, we can rewrite this as follows:
$$\left[ \begin{array}{c}9\\5\\6\\10\end{array}\right] = \left[ \begin{array}{c}9.5\\5.5\\5.5\\9.5\end{array}\right] + \left[ \begin{array}{c}-.5\\-.5\\.5\\.5\end{array}\right]$$
Considering this additional structure, it seems we have an outcome which primarily emphasizes the two \emph{disjoint} committees $(1_1,2_1)$ and $(1_2,2_2)$, with a slight emphasis on committees involving $1_2$ as a representative.
\end{example}

The additional structure gives us more information -- information which should be useful in analyzing anything from separability paradoxes to manipulation.

The decomposition in the previous example may seem fairly arbitrary, but it is not.  The authors of \cite{BarceloEtAl} use representation theory to canonically decompose both the space of all possible voter preferences, or \emph{profiles}, as well as that of possible outcome point totals.  There are several main results, all of which we will precisely state in the ensuing sections of this paper (or see the next subsection for a brief preview).
\begin{itemize}
\item The space of \emph{profiles of rankings} has the structure of many copies of the regular representation of the wreath product group $S_m\wr S_n$ (defined below).  Essentially, every possible structured subspace is there, many times!
\item The space of \emph{results}, or \emph{committees}, has a much tighter structure, indexed only by integers $0\leq k\leq n$, so we know that most profile information will be ignored by any procedure involving weights.
\item Just as in `ordinary' voting theory, we can construct profiles which can lead to radically different outcomes if we use different weights.
\end{itemize}

These results are a major step forward in understanding structured voting on committees.

\subsection{This Paper}
This paper builds on\footnote{And, in certain cases, fixes statements from.} \cite{BarceloEtAl} to provide more explicit information about such voting systems.  We alternate more technical results about the vector spaces in question (which may be skimmed on a first reading) with detailed information about a large number of new procedures.  In addition to this introduction and some final thoughts in Section \ref{sec:conclusion}, it proceeds as follows.

\begin{itemize}
\item In Section \ref{sec:results} we give explicit, \emph{meaningful} bases for the space of possible results.  For example, a nice analog to Saari's Borda subspace is given a good basis in Theorem \ref{thm:borda}.  In addition, we rectify the statement of the relevant theorem on the result space in \cite{BarceloEtAl}.
\item Then in Section \ref{sec:ballot} we use this information to describe a new set of procedures where the voter is only required to submit their most-preferred committee, but meaningful point totals are given to all committees.  Proposition \ref{prop:neutralpointsR} describes them all with just $n+1$ parameters, which we then give examples of.
\item In Section \ref{sec:profile} we clarify the role of orbits (or `pockets', in \cite{LeeThesis}) in the (huge) space of profiles, and its implications for creating voting paradoxes.  Example \ref{ex:2wr2allrankings} gives a clear analysis for the $n=m=2$ case.
\item Section \ref{sec:systems} uses this machinery to richly expand the number of systems heretofore recognized, using \emph{multiple sets of weights} to create unexpectedly interesting ways to vote on structured committees.  Proposition \ref{prop:numberconstants} gives the full variety of $\frac{m^n!m^n}{(m!)^n n!}$ parameters, which we then fully analyze for the case $n=m=2$ in Subsections \ref{subsec:systems2wr2same} and \ref{subsec:systems2wr2different}.
\item We append more involved proofs in Appendix \ref{sec:proofs}.
\end{itemize}

We hope this paper will be another encouragement to researchers -- and practitioners -- to consider structure in multi-winner voting.

\section{The Committee Space}\label{sec:results}

\subsection{Introduction}
Given the structure already mentioned, there are $m^n$ total possible committees, which we notate as follows.

\begin{notation}
Suppose there are $n$ departments, each with $m$ options for representation in a committee.  Committees are denoted as $(1_{j_1},2_{j_2},\ldots n_{j_n})$, where $i_{j_i}$ means the member selected for this committee from the $i$th department is the $j_i$th option for that department's representative.  We order these lexicographically. (See Section 2.1 of \cite{BarceloEtAl}.)  The set of all such committees may be denoted $\mathcal{C}_{m,n}$ for convenience.
\end{notation}

\begin{definition}
The finite-dimensional $\mathbb{Q}$-vector space with a natural basis indexed by the elements of $\mathcal{C}_{m,n}$ is called the \emph{committee space}, or the \emph{results space} in certain contexts.  If need be, we will call this space $R=R_{m,n}$ (because $C$ and $c$ will stand for committees).
\end{definition}

Our main goal in this section is to have a systematic way to decompose vectors in this space meaningful in terms of the committee structure.

\begin{example}\label{ex:2wr2third}
We can further decompose the vector in Example \ref{ex:2wr2second} as a sum of three mutually orthogonal vectors
$$\left[ \begin{array}{c}9\\5\\6\\10\end{array}\right] = \left[ \begin{array}{c}7.5\\7.5\\7.5\\7.5\end{array}\right] + \left[ \begin{array}{c}-.5\\-.5\\.5\\.5\end{array}\right] + \left[ \begin{array}{c}2.5\\-2.5\\-2.5\\2.5\end{array}\right] \quad\begin{array}{c}(1_1,2_1) \\(1_1,2_2) \\(1_2,2_1) \\(1_2,2_2) \end{array}\; .$$
Note from the committee labels on the right that if we relabel the committees in any systematic way (such as $1_1\leftrightarrow 1_2$), these summands at most change sign on an individual basis, while the original vector would change substantially.
\end{example}

In examining this example, it should be clear there is no particular reason it has to indicate scores from a voting method.   We could just as easily be examining a set of thirty voters, each of which has a favorite committee out of the four options.  Under this interpretation, one third of the voters think $1_2$ and $2_2$ should be working together.  This explains why we can call this vector space either a results space or a committee space; Section \ref{sec:ballot} will pursue this idea.

More immediately, it seems plausible that our initial ad-hoc decomposition may be able to be refined meaningfully, and hopefully systematically.  The rest of this section will unpack this fully, building throughout on \cite{BarceloEtAl}.

\subsection{Representations}\label{subsection:reps}

To exploit the structure in our committees, we will need a symmetry group. While there are $m^n$ committees, the additional structure evident by department suggests we should not consider all $(m^n)!$ different permutations of all committees, but only the ones corresponding a combination of permutations \emph{of} departments combined with permutations \emph{within} departments.  Recall the following (\cite{BarceloEtAl} (2.2)):

\begin{definition}
The \emph{wreath product} $S_m\wr S_n$ is given by the set of all pairs $(\sigma;\pi)$ where $\pi\in S_n$ and $\sigma=(\sigma_1,\sigma_2,\ldots ,\sigma_n)\in S_m^n$.  The group structure is given as expected for $S_n$, but where $\pi$ permutes the elements of $\sigma'$ if we multiply $(\sigma;\pi)$ by $(\sigma';\pi')$.
\end{definition}

Most important for this paper is that there is an action of $S_m\wr S_n$ on $\mathcal{C}_{m,n}$, given by sending $(1_{j_1},2_{j_2},\ldots ,n_{j_n})$ to $(1_{\sigma_1(j_{\pi^{-1}(1)})},2_{\sigma_2(j_{\pi^{-1}(2)})},\ldots, n_{\sigma_n(j_{\pi^{-1}(n)})})$.  This is notationally dense, but is the same as first permuting the departments by $\pi$ and then permuting within each department $i$ by $\sigma_i$.  The relabeling mentioned in Example \ref{ex:2wr2third} is the action of the element  $(((12),e);e)$.

This action on $\mathcal{C}_{m,n}$ immediately extends to an action on $R_{m,n}$ by permuting the entries by basis vector, which clearly commutes with multiplication by scalars in $\mathbb{Q}$.  This gives $R$ the structure of a (permutation) \emph{representation} over $\mathbb{Q}$, or $\mathbb{Q}S_m\wr S_n$-module\footnote{See e.g.~\cite{Sagan} for basic information on representations.}.

Now recall that an \emph{irreducible} representation is a vector space with compatible group (here, $S_m\wr S_n$) action such that no proper nontrivial subspace is closed under the group action.  One of the wonderful facts about representations is that, given the (finite) symmetry group, there is always a decomposition of any representation $V$ as a direct sum of copies of a \emph{finite} number of isomorphism classes of irreducible representations -- and this can often be explicitly computed.

To do so in this case, recall that a \emph{partition} $\lambda$ of $m$, or $\lambda \vdash m$, is a (weakly) ordered set of positive integers summing to $m$.   For $S_m\wr S_n$, \cite{JamesKerber} tells\footnote{See section 4.3.34, sections 4.4.1-4.4.3, explicitly so on page 154, with computations for $S_2\wr S_3$ on pages 156-7.} us that the irreducible representations of $S_m\wr S_n$ are indexed $S^{\boldsymbol{\lambda}}$, where $\boldsymbol{\lambda}=(\lambda_1,\lambda_2,\ldots ,\lambda_{p(m)})$.  Here, $p(m)$ is the partition counting function (so the entries are ordered by the lexicographic ordering of partitions of $m$ as well), and $\lambda_j\vdash \ell_j$ where $\sum_{j=1}^{p(m)} \ell_j=n$ (and where if $\ell_j=0$ then the partition is necessarily $\emptyset$).

In general a space indexed by such `partition of partitions' could be quite complex, but for $R_{m,n}$ \cite{BarceloEtAl} shows the committee space is much simpler.

\begin{theorem}\label{Correct3.8}
The committee/results space $R_{m,n}$ has the following decomposition into simple $\mathbb{Q}S_m\wr S_n$-submodules: $$R_{m,n}\simeq \bigoplus_{k=0}^n S^{((n-k),(k),\emptyset,\ldots,\emptyset)}\, .$$  Each such submodule has dimension $\binom{n}{k}(m-1)^k$, and is in fact a direct sum of $\binom{n}{k}$ $S_m^n$-isomorphic (but not equal) modules of the form $S^{(m)^{\otimes^{n-k}}}\otimes S^{(m-1,1)^{\otimes^k}}$.
\end{theorem}

This is essentially \cite{BarceloEtAl}, Theorem 3.8, which had an error in its statement and proof; see \ref{proof:Correct3.8} below for the proof.  We have already used this in Examples \ref{ex:2wr2first}, \ref{ex:2wr2second}, and \ref{ex:2wr2third}, where the spaces are indexed by $((2),\emptyset)$, $((1),(1))$, and $(\emptyset,(2))$.

\subsection{Explicit Construction of Useful Bases}

In order to use these spaces concretely, the remainder of this section will give as much information as possible about a `useful' basis for each of the irreducible subspaces in question.  To begin, here are two conventions.

\begin{itemize}
\item We use the lexicographic order on committees of the form $(1_{j_1},2_{j_2},\ldots , n_{j_n})$ throughout for the entries in our (column) vectors in $\mathbb{Q}^{m^n}$.
\item The natural representation $\mathcal{N}\simeq S^{(m)}\oplus S^{(m-1,1)}$ has as a convenient basis the (trivial space) vector $[1,1,\cdots ,1]^T$ and any $m-1$ vectors from the set
$$\left\langle  \left[\begin{array}{c}m-1\\-1\\-1\\ \vdots\\-1\end{array}\right], \left[\begin{array}{c}-1\\m-1\\-1\\ \vdots\\-1\end{array}\right], \cdots ,\left[\begin{array}{c}-1\\-1\\ \vdots\\-1\\m-1\end{array}\right]\right\rangle$$ 
where we will typically take the first $m-1$ of these.
\end{itemize}

\begin{fact}\label{fact:kronecker}
If we have two bases $\mathcal{B}_1,\mathcal{B}_2$ of $d_1,d_2$-dimensional subspaces $V_1,V_2$ of $\mathbb{Q}^n$, we can consider their column matrices $M_{\mathcal{B}_1},M_{\mathcal{B}_2}$.  The \emph{Kronecker product} of these matrices, where each element of $M_{\mathcal{B}_1}$ is multiplied by each element of $M_{\mathcal{B}_2}$, gives a basis for the vector space $V_1\otimes V_2$ as a subspace of $\mathbb{Q}^{n^2}$.
\end{fact}

By construction in Theorem \ref{Correct3.8} and its proof, we may build up the spaces by doing the Kronecker product of bases for $S^{(m)}$ and $S^{(m-1,1)}$ in various combinations.  Here is an example.

\begin{example}\label{ex:3wr3tensor21}
Let $m,n=3$ (three departments, three candidates per department), and let $k=1$.  Then the process in Fact \ref{fact:kronecker} gives the following bases for the three isomorphic (but not equal) spaces $S^{(3)}\otimes S^{(3)}\otimes S^{(2,1)}$, $S^{(3)}\otimes S^{(2,1)}\otimes S^{(3)}$, and $S^{(2,1)}\otimes S^{(3)}\otimes S^{(3)}$.

\scriptsize
$$\left\langle
\left[\begin{array}{c}2 \\-1 \\-1 \\2 \\-1 \\-1 \\2 \\-1 \\-1 \\2 \\-1 \\-1 \\2 \\-1 \\-1 \\2 \\-1 \\-1 \\2 \\-1 \\-1 \\2 \\-1 \\-1 \\2 \\-1 \\-1\end{array}\right],
\left[\begin{array}{c}-1 \\2 \\-1 \\-1 \\2 \\-1 \\-1 \\2 \\-1 \\-1 \\2 \\-1 \\-1 \\2 \\-1 \\-1 \\2 \\-1 \\-1 \\2 \\-1 \\-1 \\2 \\-1 \\-1 \\2 \\-1\end{array}\right]
\right\rangle, \left\langle
\left[\begin{array}{c}2 \\2 \\2 \\-1 \\-1 \\-1 \\-1 \\-1 \\-1 \\2 \\2 \\2 \\-1 \\-1 \\-1 \\-1 \\-1 \\-1 \\2 \\2 \\2 \\-1 \\-1 \\-1 \\-1 \\-1 \\-1\end{array}\right],
\left[\begin{array}{c}-1 \\-1 \\-1 \\2 \\2 \\2 \\-1 \\-1 \\-1 \\-1 \\-1 \\-1 \\2 \\2 \\2 \\-1 \\-1 \\-1 \\-1 \\-1 \\-1 \\2 \\2 \\2 \\-1 \\-1 \\-1\end{array}\right]
\right\rangle, \left\langle
\left[\begin{array}{c}2 \\2 \\2 \\2 \\2 \\2 \\2 \\2 \\2 \\-1 \\-1 \\-1 \\-1 \\-1 \\-1 \\-1 \\-1 \\-1 \\-1 \\-1 \\-1 \\-1 \\-1 \\-1 \\-1 \\-1 \\-1\end{array}\right],
\left[\begin{array}{c}-1 \\-1 \\-1 \\-1 \\-1 \\-1 \\-1 \\-1 \\-1 \\2 \\2 \\2 \\2 \\2 \\2 \\2 \\2 \\2 \\-1 \\-1 \\-1 \\-1 \\-1 \\-1 \\-1 \\-1 \\-1\end{array}\right]
\right\rangle
\quad\begin{array}{c} (1_1,2_1,3_1) \\(1_1,2_1,3_2) \\(1_1,2_1,3_3)\\ (1_1,2_2,3_1) \\(1_1,2_2,3_3) \\(1_1,2_2,3_3)\\ (1_1,2_3,3_1) \\(1_1,2_3,3_3) \\(1_1,2_3,3_3)\\
 (1_2,2_1,3_1) \\(1_2,2_1,3_2) \\(1_2,2_1,3_3)\\ (1_2,2_2,3_1) \\(1_2,2_2,3_3) \\(1_2,2_2,3_3)\\ (1_2,2_3,3_1) \\(1_2,2_3,3_3) \\(1_2,2_3,3_3)\\
 (1_3,2_1,3_1) \\(1_3,2_1,3_2) \\(1_3,2_1,3_3)\\ (1_3,2_2,3_1) \\(1_3,2_2,3_3) \\(1_3,2_2,3_3)\\ (1_3,2_3,3_1) \\(1_3,2_3,3_3) \\(1_3,2_3,3_3)  \end{array}$$
\normalsize

The direct sum of the three spaces generated by these bases clearly will now have an $S_n=S_3$-action as well, so that larger space will be $S^{((2),(1))}$ by our earlier construction.

\end{example}

The basis in Example \ref{ex:3wr3tensor21} is interesting -- for each department $j$, we have a vector with $2$ points for a committee with $j_1$ (and $-1$ points otherwise) and similarly for $j_2$.    By construction the following is easy to see.

\begin{prop}\label{prop:kronecker}
Choose a particular ordering of the tensors of $S^{(m)^{\otimes^{n-k}}}\otimes S^{(m-1,1)^{\otimes^k}}$ corresponding to choosing $k$ of the $n$ departments.   For each possible subcommittee $c_k$ for these departments, we can use the construction in Fact \ref{fact:kronecker} (with $m-1$ in the position of the candidate for department $i$ for $c_k$) to construct a basis vector $b_{c_k}$ in (that ordering of) $S^{(m)^{\otimes^{n-k}}}\otimes S^{(m-1,1)^{\otimes^k}}$.  The value of an entry will be $(m-1)^\ell (-1)^{k-\ell}(1)^{n-k}=\pm (m-1)^{\ell}$ precisely if it agrees with the subcommittee $c_k$ in exactly $\ell$ of the $k$ departments in question.

Doing so for all orderings of $S^{(m)^{\otimes^{n-k}}}\otimes S^{(m-1,1)^{\otimes^k}}$ will yield a basis for $S^{((n-k),(k),\emptyset,\ldots,\emptyset)}$.  
\end{prop}

\begin{remark}\label{rmk:anywilldo}
We observe that in Fact \ref{fact:kronecker} and Proposition \ref{prop:kronecker}, only $m-1$ of the candidates from each department will show up in the bases as constructed.  However, the spanning set of $m$ such vectors is useful, and they sum to zero.  This is even true when $k>1$, if we hold all but one department fixed -- the Kronecker construction cannot undo a dependence relation.
\end{remark}

\begin{example}
If we had done the same construction for $k=2$ in Example \ref{ex:3wr3tensor21}, we would have gotten $4$ points for any committee sharing $j_2$ and $j_3$, $-2$ points for any including just one of them, and $1$ point for those sharing none.
\end{example}

We will have opportunity to use the same type of construction often, so we introduce some notation.

\begin{notation}
Given $n$, we denote an element $\p$ of $R_{m,n}$ such that there exists a committee $C$ where the value of $\p[C']$ depends strictly on the cardinality of $C\cap C'$ by $\p = [p_0;p_1;p_2;\ldots ;p_n]$, where $C$ should be clear by context.
\end{notation}

The previous example then may be notated $[4;-2;1]$.

Now we use the preceding work to get a much more interesting basis for $S^{((n-k),(k),\emptyset,\ldots ,\emptyset)}$.  See \ref{subsec:resultsproofs} for the proof.

\begin{prop}\label{prop:committeek}
Select a specific full committee $C$ and a choice of compatible subcommittees $c_k$ of size $k$ (that is, each $c_k\subset C$ as sets), one for each set of $k$ of the $n$ departments.  Then choose the $\binom{n}{k}$ vectors $b_{c_k}$ from each $S^{(m)^{\otimes^{n-k}}}\otimes S^{(m-1,1)^{\otimes^k}}$ in Proposition \ref{prop:kronecker}.

Their \emph{sum} $b_C$ is a vector in $\mathbb{Q}^{m^n}$ which corresponds to the committee $C$ in the following sense.  For each entry in the vector $b$ corresponding to the committee $C'=(1_{j'_1},2_{j'_2},\ldots ,n_{j'_n})$, its value is determined solely by the cardinality of $C\cap C'$, and these values depend solely on $m,n,k$, not on $C$.
\end{prop}

\begin{remark}\label{rmk:structures}
The similarity to the table in Section 4.2 of \cite{CrismanPermuta}, where various pieces of the space of all profiles on linear orders decompose with entries depending only the rank of a candidate in each order, should be evident.  See also \cite{SaariStruct1}, \cite{SaariStruct2}.
\end{remark}

A simple induction argument\footnote{For example, the one by \cite{4453618} at {\tt https://math.stackexchange.com/a/4453618/24113}.} using Remark \ref{rmk:anywilldo} yields the following important result.

\begin{prop}\label{prop:committeespan}
The set of all such $b_C$, including the ones mentioned in Remark \ref{rmk:anywilldo}, spans $S^{((n-k),(k),\emptyset,\ldots,\emptyset)}$.
\end{prop}

\begin{example}\label{ex:3wr3tensor21rewrite}
Consider the same setup as in Example \ref{ex:3wr3tensor21} ($m,n=3$ and $k=1$).   Using\footnote{If the vectors in Example \ref{ex:3wr3tensor21} are labeled $b_1,b_2,c_1,c_2,d_1,d_2$, then this basis is $b_1+c_1+d_1$, $b_2+c_1+d_1$, $-(b_1+b_2)+c_1+d_1$, $b_1+c_2+d_1$, $b_1+c_2+d_2$, and $b_1-(c_1+c_2)+d_1$, corresponding to the committees $(1_1,2_1,3_1)$, $(1_2,2_1,3_1)$, $(1_3,2_1,3_1)$, $(1_1,2_2,3_1)$, $(1_1,2_2,3_2)$, and $(1_1,2_3,3_1)$.} Proposition \ref{prop:committeespan}, here is a \emph{different} meaningful basis for $S^{((2),(1),\emptyset)}$.  

\scriptsize
$$\left\langle
\left[\begin{array}{c}6\\3\\3\\3\\0\\0\\3\\0\\0\\3\\0\\0\\0\\-3\\-3\\0\\-3\\-3\\3\\0\\0\\0\\-3\\-3\\0\\-3\\-3\end{array}\right],
\left[\begin{array}{c}3\\6\\3\\0\\3\\0\\0\\3\\0\\0\\3\\0\\-3\\0\\-3\\-3\\0\\-3\\0\\3\\0\\-3\\0\\-3\\-3\\0\\-3\end{array}\right],
\left[\begin{array}{c}3\\3\\6\\0\\0\\3\\0\\0\\3\\0\\0\\3\\-3\\-3\\0\\-3\\-3\\0\\0\\0\\3\\-3\\-3\\0\\-3\\-3\\0\end{array}\right],
\left[\begin{array}{c}3\\0\\0\\6\\3\\3\\3\\0\\0\\0\\-3\\-3\\3\\0\\0\\0\\-3\\-3\\0\\-3\\-3\\3\\0\\0\\0\\-3\\-3\end{array}\right],
\left[\begin{array}{c}0\\-3\\-3\\3\\0\\0\\0\\-3\\-3\\3\\0\\0\\6\\3\\3\\3\\0\\0\\0\\-3\\-3\\3\\0\\0\\0\\-3\\-3\end{array}\right],
\left[\begin{array}{c}3\\0\\0\\3\\0\\0\\6\\3\\3\\0\\-3\\-3\\0\\-3\\-3\\3\\0\\0\\0\\-3\\-3\\0\\-3\\-3\\3\\0\\0\end{array}\right]
\right\rangle
\quad\begin{array}{c} (1_1,2_1,3_1) \\(1_1,2_1,3_2) \\(1_1,2_1,3_3)\\ (1_1,2_2,3_1) \\(1_1,2_2,3_3) \\(1_1,2_2,3_3)\\ (1_1,2_3,3_1) \\(1_1,2_3,3_3) \\(1_1,2_3,3_3)\\
 (1_2,2_1,3_1) \\(1_2,2_1,3_2) \\(1_2,2_1,3_3)\\ (1_2,2_2,3_1) \\(1_2,2_2,3_3) \\(1_2,2_2,3_3)\\ (1_2,2_3,3_1) \\(1_2,2_3,3_3) \\(1_2,2_3,3_3)\\
 (1_3,2_1,3_1) \\(1_3,2_1,3_2) \\(1_3,2_1,3_3)\\ (1_3,2_2,3_1) \\(1_3,2_2,3_3) \\(1_3,2_2,3_3)\\ (1_3,2_3,3_1) \\(1_3,2_3,3_3) \\(1_3,2_3,3_3)  \end{array}$$
\normalsize

Note that each vector `awards' $6$ points to a `target' committee, $3$ points to any committee sharing two candidates with this committee, $0$ points to one sharing only one member, and $-3$ points to any committee disjoint to the `target'.  So the vectors above have the form $[6;3;0;-3]$.

\end{example}

\subsection{The Committee Subspaces}
In Example \ref{ex:3wr3tensor21rewrite}, we saw that $S^{((2),(1),\emptyset)}$ has a particularly nice basis, where the entries of a given vector seem to be a linear function of how many candidates are shared.  This is not a coincidence.

\begin{theorem}\label{thm:borda}
In Proposition \ref{prop:committeespan}, let $b_C$ be the vote/points vector corresponding to a committee $C$.  Suppose also $k=1$ (so, we are choosing just one department at a time).  The values of the entry for a committee $C'$ in $b_C$ is $(m-1)(n-d)-d=mn-n-md$ precisely when the committees \emph{disagree} in $|C\cap C'|=d$ departments, or $[mn-n;mn-n-m;\ldots ;-n]$.
\end{theorem}

\begin{example}
When considering five candidates each in three departments ($m=5,n=3$), the values for the four possibilities would be twelve ($d=0$ or $C=C'$), seven, two, and negative three ($C$ disjoint from $C'$) points.
\end{example}

\begin{proof}[Proof of Theorem \ref{thm:borda}]
Each committee $C'$ disagrees in $d$ departments.  In that case, the vectors used to form $b_C$ will have value $-1$ in the vectors corresponding to those $d$ departments, and value $m-1$ in the other $n-d$ vectors.
\end{proof}

\begin{definition}\label{defn:Borda}
We call the $S^{((n-1),(1),\emptyset,\ldots \emptyset)}$ subspace the \emph{Borda subspace}, by analogy with the Borda subspace in \cite{SaariStruct1} (including both the algebraic structure and weights).  
\end{definition}

So Theorem \ref{thm:borda} gives a good basis for the Borda subspace.

\begin{fact}\label{fact:otherbordabasis}
There is another basis of potential interest.  Since the Borda space is a direct sum of $n$ spaces isomorphic (as $S_m$-modules) to $S^{(m-1,1)}$, call the one corresponding to the $d$th department $B_d$.  We can use Fact \ref{fact:kronecker} instead with the basis $$\left\langle  \left[\begin{array}{c}1\\-1\\0\\ \vdots\\0\end{array}\right], \left[\begin{array}{c}1\\0\\-1\\ \vdots\\0\end{array}\right], \cdots ,\left[\begin{array}{c}1\\0\\ \vdots\\0\\-1\end{array}\right]\right\rangle\; .$$   This gives vectors $v_{d,i}$ such that, for a given $d$, with entry $1$ for any committee including $d_1$, entry $-1$ including $d_i$, and zero otherwise.
\end{fact}

The Borda space (by Theorem \ref{thm:borda}) expresses, along with the trivial space, the set of all possible points values/profiles with \emph{linear} relationships between the numbers and number of disagreements between committees.

Likewise, each of the $S^{((n-k),(k),\emptyset,\ldots \emptyset)}$ irreducible subspaces in the decomposition of the committee space corresponds to a particular size ($k$) subcommittee in Proposition \ref{prop:kronecker}.  Some of these spaces are useful to delineate further.

\begin{definition}\label{defn:sign}
We call the $S^{(\emptyset,(n),\emptyset,\ldots \emptyset)}$ subspace the \emph{sign subspace}, since the values in Proposition \ref{prop:kronecker} are largely determined by the parity of $d$ the number of departments that disagree with a target committee, $(m-1)^\ell (-1)^{k-\ell}=(-1)^d (m-1)^{n-d}$.
\end{definition}

The sign subspace really does just give sign/parity in the important case where there are only two candidates per department.

% Include only if I can't prove the general one
\begin{prop}\label{prop:level2component}
For the $S^{((n-2),(2),\emptyset,\ldots ,\emptyset)}$ component, the weights in a basis vector from Proposition \ref{prop:committeespan} for a committee disagreeing in $d$ spots from a target committee is $$\frac{m^2}{2}(d^2-d)-m(m-1)(n-1)d+\binom{n}{2}(m-1)^2$$ or $$\frac{m^2}{2}d^2+\left(-m(m-1)(n-1)-\frac{m^2}{2}\right)d+\frac{n(n-1)}{2}(m-1)^2\, .$$
\end{prop}

See \ref{subsec:resultsproofs} for the proof.  This component, roughly speaking, gives strong weight to committees either sharing \emph{or} not sharing many departments in common with a target committee. (The `roughly' is because there are so many more committees that do \emph{not} share candidates once $m$ is greater than two.)

\begin{example}
When considering five candidates each in three departments ($m=5,n=3$), the values would be $48$ ($d=0$ or $C=C'$), $8$, $-7$, and $3$, yielding $[48;8;-7;3]$.
\end{example}

% DO THESE IF WE THINK THEY CAN FIT

%\begin{prop}
%In fact, in the $j$th component, we would see that the number who disagree in one department changes by $(m-1)^{j-1}m$ in $\binom{n-1}{j-1}$ of the $j$-tuple subcommittees, so the first two terms of each polynomial should be $$\binom{n}{j}(m-1)^j-m(m-1)^{j-1}\binom{n-1}{j-1}t$$
%\end{prop}
%
%More generally, I think I can conjecture that the differences for $m,n$ as always and as the $j$ basis component and the degree $k$ difference is $$(-1)^k(m-1)^{j-k}m^k\binom{n-k}{j-k}$$ so that the polynomial would then be $$\sum_{k=0}^{j} (-1)^k\frac{t_{(k)}}{k!}(m-1)^{j-k}m^k\binom{n-k}{j-k}$$ where $t_{(k)}=t(t-1)(t-2)\cdots (t-k+1)$.

The point of all this computation, of course, is to find structure in profiles of voters that we might not suspect.  We can use that in Section \ref{sec:ballot}.

\begin{example}
Pick a committee $C$ where $m=n=3$.  Consider a profile with eight voters preferring $C$, but two voters for each committee disjoint from $C$ and one voter for each committee sharing only one committee member with $C$.  Then this profile has \emph{no} Borda component, which is to say that this voter profile is orthogonal to \emph{all} linear relationships involving how many candidates are shared with a given committee.
\end{example}

\section{The Committee as Ballot}\label{sec:ballot}

In situations with a fair amount of structure in the set of objects voted on, it is reasonable to ask for voters to use a \emph{plurality ballot} rather than to rank all possible outcomes, or to disaggregate in other ways.  While in Section \ref{sec:systems} we will look at more complex systems, in the current section we will focus on voting systems which have votes that are simply elements of $\mathcal{C}_{m,n}$.

\begin{definition}\label{defn:points-based-general}
Given $m,n$ let $s:\mathcal{C}_{m,n}\times \mathcal{C}_{m,n}\to \mathbb{Q}$ be a \emph{ballot-scoring function}, which we interpret as recording how many points a vote for a given committee would give to another committee.  Then a \emph{points-based voting rule} $f$ is a function from $R_{m,n}$ to the power set of $\mathcal{C}_{m,n}$, or (by abuse of notation) $R_{m,n}\to R_{m,n}$, such that, for a profile $\p\in R_{m,n}$, $f(\p)$ is the (set of) committee(s) $c$ which maximizes $\sum_{c'\in \mathcal{C}_{m,n}} \p[c']s(c',c)$.
\end{definition}

\begin{definition}\label{defn:committeeneutral}
Further, we have already seen an action of $S_m\wr S_n$ on $\mathcal{C}_{m,n}$ and $R_{m,n}$.  Then rule $f$ is \emph{committee neutral}, or \emph{neutral}, if $s$ is $S_m\wr S_n$-invariant, in the sense that for any $(\sigma;\pi)\in S_m\wr S_n$, $s(c',c)=s((\sigma;\pi)(c'),(\sigma;\pi)(c))$.
\end{definition}

Note that this is \emph{not} the same as the usual neutrality for plurality ballots, where all ballots are treated equally.  We want to take the committee structure explicitly into account, and only require neutrality among department names and individual candidates within each department.

In Section \ref{sec:results}, we saw that $R_{m,n}$ has a specific structure as a representation of $S_m\wr S_n$.  Given $(\sigma;\pi)\in S_m\wr S_n$ and $c\in \mathcal{C}_{m,n}$, if $s$ comes from a neutral rule, then
$$(\sigma;\pi)\left(\sum_{c'\in \mathcal{C}_{m,n}}\p[c'] s(c',c)\right)=\sum_{c'\in \mathcal{C}_{m,n}}\p[((\sigma;\pi)^{-1}(c')] s((\sigma;\pi))^{-1}(c'),(\sigma;\pi)^{-1}(c))$$ so a neutral points-based voting rule may be thought of as a $\mathbb{Q}S_m\wr S_n$-module homomorphism.  Thus the following applies:

\begin{schur}
If $M$ and $N$ are irreducible $\mathbb{Q}S_m\wr S_n$-modules and $g:M\to N$ is a $\mathbb{Q}S_m\wr S_n$-module homomorphism, then either $g=0$ or $g$ is an isomorphism (\cite{Sagan} Theorem 1.6.5) given by multiplication by an element of $\mathbb{Q}$ (\cite{JamesKerber} 4.4.9).
\end{schur}

\begin{prop}\label{prop:neutralpointsR}
Any neutral points-based voting rule on $\mathcal{C}_{m,n}$ is equivalent to a linear map $$\bigoplus_{k=0}^n S^{((n-k),(k),\emptyset,\ldots,\emptyset)}\to \bigoplus_{k=0}^n S^{((n-k),(k),\emptyset,\ldots,\emptyset)}$$ given by multiplication by a scalar for each $S^{((n-k),(k),\emptyset,\ldots,\emptyset)}$.

In other words, any points-based rule for $n$ departments is completely determined by $n+1$ constants.
\end{prop}

\begin{remark}
Effectively there are only $n-1$ parameters, if we only look at point differentials (killing the trivial $S^{((n),\emptyset ,\ldots \emptyset)}$) and ignore a scaling factor.
\end{remark}

\begin{remark}
By the same principle of `the acted upon has become the actor' in \cite{OrrisonSymmetry} and \cite{BarceloEtAl}, the weighting vectors also decompose in the same way as profiles in $R_{m,n}$.  More precisely, we can use the same notation $[a_0;a_1;\ldots ;a_n]$ for weights on committees disagreeing in $i$ spots, and we are guaranteed that the weights $[0;0;\ldots ;0;1;0;\ldots ;0]$ which considers \emph{only} committees that disagree in exactly $k$ spots will have the same decomposition as in the $b_C$ weights in Proposition \ref{prop:committeek}.
\end{remark}

% THIS IS THE SAME THING DIFFERENTLY EXPRESSED
%\begin{prop}
%Choose $0\leq k\leq n$.  Consider the set of weights for a procedure where, given a voter's preferred committee $c$, we grant one point to every committee that shares $n-k$ (or, differs in $k$) members with $c$, and zero to every other committee.  Then this set of weights is \emph{precisely} the weights for the basis vectors $b_C$ in Proposition \ref{prop:committeek} for the $S^{((n-k),(k),\emptyset,\ldots,\emptyset)}$ component.
%\end{prop}
% NEEDS ADDITIONAL PROOF?

Now we can talk about properties of such voting systems.  The first is essentially a dot product computation using Theorem \ref{thm:borda}.

\begin{prop}
Suppose we use a Borda-like set of points, where a committee sharing $k$ members with the chosen committee $C$ is given $k$ points, or in general $[n;n-1;\ldots ;1;0]$.    Then the parameters involved in the method are all zero, except $a_0=n m^{n-1}$ and $a_1=m^{n-1}$.  

That is to say, a Borda-like (linearly weighted) procedure ignores any profile information that is not in the Borda component or which is trivial, and it amplifies the Borda component.
\end{prop}

\begin{example}
With a $S_2\wr S_2$ vote for $(1_1,2_1)$ a Borda-like procedure would have weighting vector $[2,1,1,0]^T = [2;1;0]$.

With $S_3\wr S_2$ a vote for $(1_1,2_1,3_1)$ would give points (in the lexicographic order) as per the vector $[2,1,1,1,0,0,1,0,0]^T=[2;1;0]$ as well.
\end{example}

Given that the Borda space only has dimension $n(m-1)$, just as in regular voting theory, a Borda-like set of weights (intentionally) ignores much of the information in a profile.  If, in designing a system, you believe that the influence of a voter on the score of a particular committee should increase predictably with its intersection with the voter's favorite committee, this is a powerful feature.

This means, of course, that \emph{any system which is not Borda-like} will be emphasizing also some of the other components.   For instance, in the $S_2\wr S_2$ case, any non-linear set of weights will `stretch' the subspace generated by $[1,-1,-1,1]^T$, which favors a committee \emph{and} its opposite over against others.  

\begin{example}
Consider the voting rule in which a voter is assumed to approve of any committee not completely disjoint from their favorite committee\footnote{This may be seen as related to the `representation-focused' scoring functions of \cite{ElkindEtAlMultiwinner}, but in our structured setting.}.  This would have weights such as $[1;1;\ldots ;1;0]$.

But this is the trivial rule (every committee gets a point per voter, regardless of profile) minus the rule giving points only to completely disjoint committees, which corresponds to the sign subspace in Definition \ref{defn:sign}.  So a direct computation shows this ostensibly innocuous rule decomposes into having parameters
$$[m^n - (m-1)^n;(m-1)^{n-1};-(m-1)^{n-2};\ldots \pm 1]$$ for the various $S^{((n-k),(k),\emptyset,\ldots ,\emptyset)}$ subspaces. 

To see how this impacts the symmetries of this system, consider when we have $m=n=4$ departments and candidates.  The $k=2$ component (with typical basis profile $[54;18;-2;-6;6]$) is thus \emph{reversed} in impact by one third as much as the Borda ($[12;8;4;0;-4]$) subspace is confirmed.  So even a very clear $[66;26;2;-6;2]$ sum of these profiles in favor of a particular committee would end up with scores (up to constant) $[-18;6;14;6;-18]$.

In favor of this rule is that the huge tie among all committees agreeing in precisely two departments with the favored committee is indeed a representability result, given that a slight plurality of voters ($78$ versus $66$) desired a committee differing slightly from the most-favored one.  On the other hand, the final scores put this committee in clear last place along with all totally disjoint committees, so the society had better really value compromise if it uses this rule.
\end{example}

Depending on the preferences of the designer, emphasizing components other than Borda might even be an explicitly desired feature!

\begin{example}\label{ex:wine}
Consider the case of selecting ingredients for a meal where there is a strong, justified, a priori belief that the ingredients are `naturally' paired with other ingredients for best taste.  (Perhaps this is the case with wines and meats.)  In that case, different `voters' might have different views as to what the best pairing they desire for dinner tonight would be, but the worst possible outcome would be to have a vote that intends to \emph{mix} these pairs.  

In this scenario, a weighting vector like $[1;-1;1]$ might accurately reflect the `electorate' and its preferences \emph{about} procedures, without forcing anyone to submit a full ranking.  Or use an equivalent (in outcomes via the argmax) non-sum-zero weighting vector $[1;0;1]$; its parameters are $a_0=2$, $a_1=0$, and $a_2=2$.
\end{example}

\begin{example}
Suppose more generally that $m=2$.  Suppose we use the weighting vector $[1;0;\ldots ;0;1]$ which allocates points from a voter for $C$ to $C$ itself and to its complementary committee (which makes sense as there are two options for each department).

The surprise is that it is \emph{not} just $a_0$ and $a_n$ which are nonzero, but \emph{all even} $a_{2k}=2$ and all odd parameters are zero.  So not only is the Borda component killed, but so are all components which fix an odd number of committees to agree with -- while the components corresponding to even-cardinality subcommittees are essentially kept unchanged!
\end{example}

\begin{proof}
There is an action of $\mathbb{Z}_2$ on $R_{2,n}$ given by sending a committee to its complement (and extending to the vector space).  Since the irreducible spaces are formed by taking tensor product powers of $$\mathcal{N}\simeq \left\langle\left[ \begin{array}{c}1\\1\end{array}\right]\right\rangle\oplus \left\langle\left[ \begin{array}{c}1\\-1\end{array}\right]\right\rangle$$ we see that the odd powers would have to be multiplied by $-1$ if they survived such an action (and hence they do not).  Further, since this set of weights gives one point to a committee \emph{and} its complement, that is exactly twice as big as the original vectors coming from $\left[ \begin{array}{c}1\\1\end{array}\right]$.
\end{proof}

\begin{example}
There are other parity criteria one could imagine.  When $m=2$, we could try the weights $[1;0;1;0;\ldots]$ where we only give points to committees which disagree with the most-preferred committee in an \emph{even} number of departments.  (This could be to structurally embed in the system a value for having \emph{pairs} of representatives being valued, but where any pair would do.)  A similar argument to the previous example shows that this kills all but the sign and trivial components, as one would expect for $m=2$.
\end{example}

\begin{example}
More interestingly, when $m=4$ something similar occurs.  Normalize to avoid a large $a_0$ component and consider the weights $[1;-1;1;\ldots ;\pm 1]$.  The points-based rule associated to these weights does not kill \emph{any} components, and instead up to a scaling factor alternates preserving exactly and inverting each component coming from an even or odd number of \emph{disagreements} ($a_i=(-1)^{i+n}2^n$).  

That is, if we give points based on the parity of the agreement of a committee with a vote with an even number of committees, we really scale up the components \emph{least} directly connected to individual members of the committees.  
\end{example}

As always, the real message is if we use weights which have a non-zero dot product with this set of weights, we get some of this same (undesired?) behavior.

\begin{proof}
For each committee, we know a basis element for $S^{((n-k),(k),\emptyset,\ldots ,\emptyset)}$ comes from Kronecker products of vectors of the form $[1,1,1,1]^T$ and $[3,-1,-1,-1]^T$.  At the department level, then, the specified committee from a given basis vector will receive $1-1-1-1=-2=1\cdot -2$ points from each agreeing department and $3+1+1+1=6=3\cdot 2$ from each disagreeing department.  Multiplying these together and recalling that Proposition \ref{prop:kronecker} does the same with $3$ and $-1$ means that this propagates to $S^{((n-k),(k),\emptyset,\ldots ,\emptyset)}$, with an extra factor of $2$ for each department and an extra factor of $-1$ for each agreeing department.
\end{proof}

That this happens for $m=4$, $n=2$ is very specific to the number-theoretic fact that $(m-1)n$ is an exact multiple of $2-m$, and does \emph{not} easily generalize.  There are many other combinatorial or number-theoretic facts about which components different procedures emphasize, but we have tried to emphasize some with more evident voting interest.

\section{The Profile Space}\label{sec:profile}
The papers \cite{BarceloEtAl} and \cite{LeeThesis} analyze the suggestion of \cite{RatliffPublicChoice} that each voter submit a \emph{full ranking} of all possible structured committees.  Since there are $m^n$ possible committees, there are $\left(m^n\right)!$ possible rankings!  But for artificial intelligence applications, having agents with one of this gargantuan number of strict rankings may be a reasonable assumption.

Section \ref{sec:systems} will discuss the amazing variety of voting rules of this type that can be devised, and why this variety might be of use.  However, in this section we first examine the space of all profiles and tighten up some of the results from \cite{BarceloEtAl}, in order to understand this space deeply.

\begin{definition}
The finite-dimensional $\mathbb{Q}$-vector space with a natural basis indexed by strict rankings $\mathcal{L}\left(\mathcal{C}_{m,n}\right)$ is called the \emph{profile space} $P_{m,n}$, or just $P$ if $m,n$ are clear from context.
\end{definition}

Because the $S_m\wr S_n$-action on committees permutes them, $P$ will also have a natural action by the wreath product.  Recall from Subsection \ref{subsection:reps} that the irreducible submodules under this action are indexed by partitions of partitions, $S^{\boldsymbol{\lambda}}$. 

\begin{prop}[Theorems 3.7 and 3.1 of \cite{BarceloEtAl}]\label{prop:barceloprofiledecomp}
The profile space $P_{m,n}$ has the following decomposition as irreducible $S_m\wr S_n$-modules:
$$P_{m,n} = \bigoplus_{\frac{(m^n)!}{(m!)^n n!}}\bigoplus_{\lambda\text{ irr.}}(S^{\lambda})^{\oplus \text{ dim }(S^{\lambda)}}$$
In particular, $P$ decomposes as one copy of the regular representation of $S_m\wr S_n$ for each \emph{orbit} of the set of strict rankings under the $S_m\wr S_n$-action.
\end{prop}

\begin{example}\label{ex:2wr2allrankings}
It will shortly be crucial to see how this decomposition works, so let $m=n=2$, the simplest possible case.  Following \cite{LeeThesis}, let $W=(1_1,2_1)$, $X=(1_1,2_2)$, $Y=(1_2,2_1)$, and $Z=(1_2,2_2)$.  Then in Figure \ref{figure:2wr2pockets} we list all twenty-four rankings in a very specific order.  If we consider $W\succ Y\succ X\succ Z$ as a `reference' ranking since it has disjoint committees as its first and last options, then the second orbit comes via the cycle permuting the second to fourth to third item in each ranking, while the third orbit comes from swapping the second and last item in the rankings.

\begin{figure}[H]
$$\begin{array}{lcccccccc}
\parbox{1.5in}{Orbit 1: \\Disjoint committees\\ first/last or middle} & 
\rank{W}{Y}{X}{Z} & \rank{W}{X}{Y}{Z} & \rank{Z}{Y}{X}{W} & \rank{Z}{X}{Y}{W} & \rank{Y}{Z}{W}{X}  & \rank{Y}{W}{Z}{X}  & \rank{X}{Z}{W}{Y}  & \rank{X}{W}{Z}{Y} \\  
\\
\parbox{1.5in}{Orbit 2: \\Disjoint committees\\ first/third or second/last} & 
\rank{W}{X}{Z}{Y} & \rank{W}{Y}{Z}{X} & \rank{Z}{X}{W}{Y} & \rank{Z}{Y}{W}{X} & \rank{Y}{W}{X}{Z}  & \rank{Y}{Z}{X}{W}  & \rank{X}{W}{Y}{Z}  & \rank{X}{Z}{Y}{W} \\  
\\
\parbox{1.5in}{Orbits 3: \\Disjoint committees\\ first/second or third/last} & 
\rank{W}{Z}{X}{Y} & \rank{W}{Z}{Y}{X} & \rank{Z}{W}{X}{Y} & \rank{Z}{W}{Y}{X} & \rank{Y}{X}{W}{Z}  & \rank{Y}{X}{Z}{W}  & \rank{X}{Y}{W}{Z}  & \rank{X}{Y}{Z}{W} \\  
\end{array}$$
\caption{Ranking committees by orbits}
\label{figure:2wr2pockets}
\end{figure}

Considering each of these rankings as giving a basis element for $P_{2,2}$, then the decomposition $P\simeq (\mathbb{Q}S_2\wr S_2 )^{\oplus ^3}$ is given by these three orbits; each of the subspaces then decomposes as $$\mathbb{Q}S_2\wr S_2\simeq S^{((2),\emptyset)}\oplus S^{((1),(1))^{\oplus^2}}\oplus S^{(\emptyset,(2))}\oplus S^{((1,1),\emptyset)}\oplus S^{(\emptyset,(1,1))}$$  
% Maybe I should actually explicitly give these spaces?
We will get a finer-grained description in Section \ref{sec:systems}.
\end{example}

\begin{remark}\label{rmk:profiletoresults}
The paper \cite{BarceloEtAl} consider voting rules which extend Definitions \ref{defn:points-based-general} and \ref{defn:committeeneutral} to functions from $P_{m,n}$ to $R_{m,n}$ (or the underlying power set of possible committees).  The definition of neutrality is exactly the same -- invariance under $S_m\wr S_n$.  The only difference is that our ballot-scoring function now must involve the rankings, so that $s:\mathcal{L}\left(\mathcal{C}_{m,n}\right)\times \mathcal{C}_{m,n}\to \mathbb{Q}$.  Then a points-based voting rule on rankings is simply an $S_m\wr S_n$-module homomorphism $$\bigoplus_{\frac{(m^n)!}{(m!)^n n!}}\bigoplus_{\lambda\text{ irr.}}(S^{\lambda})^{\oplus \text{ dim }(S^{\lambda)}}\to \bigoplus_{k=0}^n S^{((n-k),(k),\emptyset,\ldots,\emptyset)}\; .$$
\end{remark}

It is immediate (from Schur's Lemma) that any such rule will have any irreducible not of the form $S^{((n-k),(k),\emptyset,\ldots,\emptyset)}$ in its kernel, \emph{and} even in the $m=n=2$ case with only $\frac{2^2!}{(2!)^2 2!}=3$ orbits, that the kernel will include much isomorphic to those irreducibles as well.

In \cite{BarceloEtAl}, the case considered is where one set of $m^n$ weights is used for all rankings, regardless of orbit\footnote{In \cite{DixonMortimer}, Sections 2.6 and 2.7, the action of $S_m\wr S_n$ on our set of committees is considered as an action on the set of functions $f:[n]\to [m]$, sometimes written $[m]^{[n]}$.  This action is called the \emph{product action}, and since it faithfully permutes the committees, it gives a canonical way to embed $S_m\wr S_n$ as a subgroup of $S_{m^n}$.}.  In this case, we may consider a rule $P_{m,n}\to R_{m,n}$ as a $\mathbb{Q}S_{m^n}$-module homomorphism as well, which is the situation in \cite{OrrisonSymmetry} and `ordinary' voting theory.

A major result of regular voting theory is that, given a large number of weighting vectors and possible outcomes, there exists a \emph{single} profile which yields each of those outcomes for one of the weights.  The nominal content of Theorem 4.2 in \cite{BarceloEtAl} is that this specializes to the structured committee case.  The following is a more precise and powerful statement of that result (proof in \ref{subsec:proofsprofile}).

\begin{theorem}[Expansion of Theorem 4.2 of \cite{BarceloEtAl}]\label{thm:expandedBarceloEtAl4.2}
Let $n\geq 2$.  Suppose that ${\bf w}_1,\ldots ,{\bf w}_j$ form a set of sum-zero weighting vectors such that for each $k\in [n]$ their projections $\text{Proj}_k({\bf w}_1,\ldots ,{\bf w}_j)$ onto $S^{((n-k),(k),\emptyset,\ldots ,\emptyset)}$ are linearly independent.  Let ${\bf r}_1, \ldots ,{\bf r}_j$ be any results vectors whose entries sum to zero.  Then, for a decomposition of $P_{m,n}$ into orbits, then there exist infinitely many profiles ${\bf p}$ {\bf in each orbit} such that $T_{{\bf w}_i}({\bf p})={\bf r}_i$ for all $1\leq i\leq j$.  So the set of such profiles has $\frac{(m^n)!}{(m!)^n n!}$ parameters.
\end{theorem}

For example, in Example 4.3 of \cite{BarceloEtAl} the authors give an element of $\mathbb{Q}S_2\wr S_2$ corresponding to the first orbit in Figure \ref{figure:2wr2pockets}.  The theorem says that we should expect one of these for \emph{each} of the three orbits.  For larger $m,n$ these come in quite large numbers.

\begin{example}\label{ex:3wr3dimensions}
In the case of $S_3\wr S_3$, we can take any \emph{six} weighting vectors that stay linearly independent in the three nontrivial subspaces.  Then pick \emph{any} six results, and we get $\frac{3^3!}{(3!)^3 3!}=8401905440137617408000000$ `dimensions' of profiles that yield these specific results vectors from these specific weights!
\end{example}

The tradeoff, if we wish to stick with rules that observe the more restricted symmetry of the wreath product, is that we do not have as many possible different outcomes from a single profile as in regular voting theory, but a \emph{much} larger dimension space of such profiles to choose from.  As for the number of weights/outcomes, other than for $m=2$ a short combinatorial analysis shows that the Borda component is the smallest nontrivial component, so that in Theorem \ref{thm:expandedBarceloEtAl4.2} one may take $j=n(m-1)$.

\section{Systems upon Systems}\label{sec:systems} % Bad title?

\subsection{Taking orbits into account}

We left off in Proposition \ref{prop:barceloprofiledecomp} and Remark \ref{rmk:profiletoresults} with the observation that a (neutral) points-based voting rule on rankings is simply an $S_m\wr S_n$-module homomorphism $$\bigoplus_{\frac{(m^n)!}{(m!)^n n!}}\bigoplus_{\lambda\text{ irr.}}(S^{\lambda})^{\oplus \text{ dim }(S^{\lambda)}}\to \bigoplus_{k=0}^n S^{((n-k),(k),\emptyset,\ldots,\emptyset)}\; .$$  In the remainder of Section \ref{sec:profile} we only considered the specific case of weighting vectors on voters' rankings that are always the same.

However, there is a \emph{huge}, and very interesting, variety of voting rules available to us if we keep the decomposition of $P_{m,n}$ into orbits in mind.  Section 4.2 of \cite{LeeThesis} foresaw just how complex it could get, though without considering different weights in different orbits. 

\begin{prop}\label{prop:numberconstants}
There are $\frac{m^n!m^n}{(m!)^n n!}$ parameters involved in a \emph{general} points-based voting rule on rankings of committees.
\end{prop}
\begin{proof}
For each $k$, each results subspace $S^{((n-k),(k),\emptyset,\ldots ,\emptyset)}$ has (by Theorem \ref{Correct3.8}) dimension $\binom{n}{k}(m-1)^k$.  That means in the regular representation of $S_m\wr S_n$ there are $\binom{n}{k}(m-1)^k$ of them, so by Schur's Lemma for \emph{each} orbit there are that many parameters (for each $k$).  Apply the identity $\sum_{k=0}^n \binom{n}{k}(m-1)^k=m^n$ and multiply.
%In each of $k$ departments we choose one of $m-1$ things to disagree with from a reference committee; sum that over all possible choices and then all possible $k$ and that gives all committees.} 
\end{proof}

\begin{example}[Example \ref{ex:3wr3dimensions}, cont.]
Even in the case of $m=3$ candidates for each of $n=3$ departments, we would have an astonishing $226851446883715670016000000$ parameters to divide up among the over $8$ septillion orbits.
\end{example}

The interplay between these analyses can yield much fruit.  By regular voting theory (e.g.~\cite{OrrisonSymmetry}) we then are analyzing functions $P\simeq \mathbb{Q}S_{m^n}\to \mathbb{Q}^{m^n}$, and for any set of nontrivial weights there is a profile that can yield any particular outcome vector (by Schur's Lemma and the decomposition $\mathbb{Q}^{m^n}\simeq S^{(m^n)}\oplus S^{(m^n-1,1)}$).  This is also true of the Borda count, but in the committee setting there is a useful contrast!

\begin{prop}\label{prop:bordaweightsborda}
The Borda weights $[m^n-1,m^n-2,\ldots 3,2,1,0]^T$, considered as an element of $R$, are in the trivial and Borda subspaces.
\end{prop}

(Proof in \ref{subsec:proofssystems}.)

\begin{prop}\label{prop:bordaweightsaction}
Ordinarily, the Borda weight used uniformly would have an image of the entirety of $R$.  However, the previous result shows that if we use these weights \emph{under the action of $S_{m^n}$ in each orbit}, the image will lie only in $S^{((n),\emptyset,\ldots ,\emptyset)}\oplus S^{((n-1),(1),\emptyset,\ldots ,\emptyset)}$.
\end{prop}

We will explain `under the action' by example in Remark \ref{rmk:2wr2identdecomp}.

\subsection{Weights and spaces for the case $S_2\wr S_2$}\label{subsec:systems2wr2same}

The number of possible rules is too big to visualize in any case beyond $m=n=2$, but even here the diversity of methods is quickly apparent, and worth exploring.  

Let $m=n=2$.  Recall the three orbits involving $W=(1_1,2_1)$, $X=(1_1,2_2)$, $Y=(1_2,2_1)$, and $Z=(1_2,2_2)$ in Figure \ref{figure:2wr2pockets}.  Recall also that we can think of $P_{2,2}$ as being a vector space generated by a basis with elements indexed by the twenty-four rankings, and with $S_2\wr S_2$-action induced by the action on the actual committees.  By Proposition \ref{prop:numberconstants} there should be twelve parameters for a given voting rule.

These parameters are grouped by orbit type. 
\begin{itemize}
\item For rankings where the first and last place committees are disjoint (such as $W\succ X\succ Y\succ Z$), we can give $a_1$ points to the first-place committee, $b_1$ points to second-place, $c_1$ and $d_1$ for the last two spots.
\item If a ranking is in the second orbit where the first and third place committees are disjoint (such as $W\succ X\succ Z\succ Y$), call the weights $a_2$, $b_2$, etc.
\item Finally, if for some reason a voter put disjoint committees in the first two spots (like $W\succ Z\succ X\succ Y$), call them $a_3$, $b_3$, etc.
\end{itemize}

Hence we have three sets of four weights, for twelve total.  It is extremely important to emphasize that \emph{there is absolutely no requirement that any of these weights be connected with each other.}  In principle, we can have completely different weights for rankings with different structures -- in the $m=n=2$ case, where they place disjoint committees in different relationships to each other.  That said, to obtain \emph{one} set of weights ${\bf w}$ as in \cite{BarceloEtAl} and \cite{LeeThesis}, we simply need to let all $a_i=a$ and so forth.

Since a voting rule can be thought of as a linear transformation, there should be a matrix for each rule, corresponding to the scoring function $s$.  For most of the maps in this paper we omit the matrices for length and size consideration, but it is worth including Figure \ref{figure:2wr2weights} where we give the complete matrix of weights for a voting rule, \emph{organized by orbits}.   For example, if a voter chose a ranking where the first and last choice were \emph{not} disjoint but of the form $W\succ X\succ Z\succ Y$, $W$ would receive $a_2$ points, $X$ would receive $b_2$ points, and so forth.

\begin{figure}[H] % Thank you Sage for this!
%\centering
\begin{adjustbox}{center}$\left(\begin{array}{rrrrrrrr|rrrrrrrr|rrrrrrrr}
a_{1} & a_{1} & d_{1} & d_{1} & c_{1} & b_{1} & c_{1} & b_{1} & a_{2} & a_{2} & c_{2} & c_{2} & b_{2} & d_{2} & b_{2} & d_{2} & a_{3} & a_{3} & b_{3} & b_{3} & c_{3} & d_{3} & c_{3} & d_{3} \\
c_{1} & b_{1} & c_{1} & b_{1} & d_{1} & d_{1} & a_{1} & a_{1} & b_{2} & d_{2} & b_{2} & d_{2} & c_{2} & c_{2} & a_{2} & a_{2} & c_{3} & d_{3} & c_{3} & d_{3} & b_{3} & b_{3} & a_{3} & a_{3} \\
b_{1} & c_{1} & b_{1} & c_{1} & a_{1} & a_{1} & d_{1} & d_{1} & d_{2} & b_{2} & d_{2} & b_{2} & a_{2} & a_{2} & c_{2} & c_{2} & d_{3} & c_{3} & d_{3} & c_{3} & a_{3} & a_{3} & b_{3} & b_{3} \\
d_{1} & d_{1} & a_{1} & a_{1} & b_{1} & c_{1} & b_{1} & c_{1} & c_{2} & c_{2} & a_{2} & a_{2} & d_{2} & b_{2} & d_{2} & b_{2} & b_{3} & b_{3} & a_{3} & a_{3} & d_{3} & c_{3} & d_{3} & c_{3}
\end{array}\right)$\end{adjustbox}
\caption{Matrix for the three orbits}
\label{figure:2wr2weights}
\end{figure}

It is now much easier to see where everything goes.  We have three copies of $\mathbb{Q}S_2\wr S_2$, each of which decomposes as $$\mathbb{Q}S_2\wr S_2\simeq S^{((2),\emptyset)}\oplus S^{((1),(1))^{\oplus^2}}\oplus S^{(\emptyset,(2))}\oplus S^{((1,1),\emptyset)}\oplus S^{(\emptyset,(1,1))}$$  
% Maybe I should actually explicitly give these spaces?
We know the latter two (one-dimensional) subspaces will go to zero since they do not show up in the decomposition of $R$, but there are four other subspaces and hence four other weights in each orbit.  For each of the three orbits, the following decomposition of the weights is most useful.

$$\begin{bmatrix}a\\b\\c\\d\end{bmatrix}=x_1\begin{bmatrix}1\\1\\1\\1\end{bmatrix}+x_2\begin{bmatrix}1\\-1\\-1\\1\end{bmatrix}+x_3\begin{bmatrix}1\\0\\0\\-1\end{bmatrix}+x_4\begin{bmatrix}0\\1\\-1\\0\end{bmatrix}$$

The following proposition clarifies exactly which subspaces $S^{((2),\emptyset)}\oplus S^{((1),(1))^{\oplus^2}}\oplus S^{(\emptyset,(2))}$ are in the kernel of our voting rule in each orbit (considering the three orbits as isomorphic), as well as the exact subspace which is the orthogonal complement to the kernel, the \emph{effective space} of \cite{OrrisonSymmetry}, which by Schur's Lemma must have the same algebraic structure.  We say \emph{governs} if a parameter (including one from Schur's Lemma) is directly proportional to another parameter.

\begin{prop}\label{prop:2wr2forpockets}
The $S^{((2),\emptyset)}$ subspace in all three orbits is governed by $x_1$, which is governed by $a+b+c+d$.  If this is zero, then it is killed, otherwise not. The behavior of the other spaces is more complex.  

In the first orbit:
\begin{itemize}
\item The parameter $x_2$ is governed by whether or not $a+d=b+c$ (if so, $x_2$ is zero).  It governs $S^{(\emptyset,(2))}$. 
\item The parameters $x_3$, $x_4$ in both cases work with $S^{((1),(1))}$ subspaces, and their precise combination determines which precise two-dimensional subspace is in the effective space (and its complement in $S^{((1),(1))^{\oplus^2}}$ is in the kernel).
\end{itemize}
In the second orbit:
\begin{itemize}
\item The parameter $x_2$ in the other two orbits will go with the weight $[1,1,-1,-1]^T$, so if it is nonzero there will be a $S^{((1),(1))}$ subspace of the effective space (and its complement in the kernel).
\item In the second orbit the weights for $x_3$ and $x_4$ become $[1,-1,0,0]^T$ and $[0,0,1,-1]^T$, so if $x_3=x_4$ it kills the respective $S^{(\emptyset,(2))}$ and if $x_3=-x_4$ the same for one of the $S^{((1),(1))}$, otherwise not.
\end{itemize}
In the third orbit this situation is the same as the second orbit, except the roles of $x_3=\pm x_4$ is reversed.
\end{prop}
\begin{proof}
This all follows by direct computation of the rank of each $4\times 8$ matrix, or in a few cases checking whether the all-ones vector is in its row space.
\end{proof}

We can verify now that if all weights are the same, as in Theorem \ref{thm:expandedBarceloEtAl4.2}, we recover computations from voting theory where all four committees are treated identically (such as the first part of Proposition \ref{prop:bordaweightsaction}).  As samples, we include another confirmation of this, and then a computation with an interesting system.

\begin{cor}
Even if $x_1=x_2=0$, it is impossible to get only some $S^{((1),(1))}$ as the effective space of the entire $T_{\bf w}$.  Moreover, it's impossible to get just $S^{((2),\emptyset)}\oplus S^{(\emptyset,(2))}$, so there is no possible way to get a two-dimensional effective space.
\end{cor}

Consider the following variant on Example \ref{ex:wine}.

\begin{example}\label{ex:oneweight}
If ${\bf w}=[1/2,-1/2,-1/2,1/2]^T$ (first and last place count) then $x_2=1/2$, so the system seems simple.  But while the first orbit has $S^{(\emptyset,(2))}$ as its effective space, the other orbits have $S^{((1),(1))}$ (in fact, \emph{identical}s, not just isomorphic, ones, given the order of the rankings) so the overall effective space is three-dimensional.  

A typical basis vector in the effective space consists of a profile with a vote for all rankings with $W$ in first or last place, and a negative vote for all others.   This rule would send such a vector to a vector in $R$ in the space spanned by $[3,-1,-1,-1]^T$, where $W$ is the only winner and the other committees tie.
\end{example}

If you don't buy the connection to Example \ref{ex:wine}, for another reason why Example \ref{ex:oneweight} is important, suppose we are examining a system using only one ${\bf w}$ where $a+d\neq b+c$, such as $a=5$, $b=4$, $c=3$, and $d=0$ which attempts to de-emphasize the last-place committee on a given ballot.  This will have a nonzero $x_2$, so the example applies to this rule as well.

\subsection{More interesting systems for $S_2\wr S_2$.}\label{subsec:systems2wr2different}

The biggest interest in these computations comes from what would happen if we did \emph{different} things in each orbit.  First, let us examine the theoretical question of whether we could avoid the results of Proposition \ref{prop:2wr2forpockets} which have each orbit behaving differently.

\begin{remark}\label{rmk:2wr2identdecomp}
We can use Proposition \ref{prop:2wr2forpockets} to make all three orbits have \emph{identical} effective spaces\footnote{Here by `identical' of course we mean identical as eight-dimensional vector spaces under the isomorphism given by identifying the $\mathbb{Q}^8$ spaces in the obvious way based on Figure \ref{figure:2wr2weights}.} by undoing the permutations which yield the orbits.
\begin{itemize}
\item For the first orbit, let $a_1,b_1,c_1,d_1=a,b,c,d$.
\item In the second orbit let $a_2=a$, but let $b_2=c$, $c_2=d$, and $d_2=b$.
\item In the third orbit, let $a_3=a$, $c_3=c$, but $b_3=d$ and $d_3=b$.
\end{itemize}
\end{remark}

\begin{example}\label{ex:2wr2allborda}
Recalling Proposition \ref{prop:bordaweightsaction}, we examine a `permuted'  Borda count, which should yield results orthogonal to $[1,-1,-1,1]^T$, which generates $S^{(\emptyset,(2))}$ in $R_{2,2}$.  For convenience we switch to the otherwise equivalent sum-zero set of weights ${\bf w}_1=[1,1/3,-1/3,-1]^T$, ${\bf w}_2=[1,-1/3,-1,1/3]^T$, and ${\bf w}_3=[1,-1,-1/3,1/3]^T$. 

Consider a profile (differential) with three voters for each of the six ballots with $W$ in first place, one for each one with $W$ in second place, negative one for $W$ in third place, and negative three for each with $W$ in last place.   Using $w_i$ as weights, we obtain a result with $W$ the lone winner and $Z$ the lone (!) \emph{loser}, with $Y$ and $X$ receiving a net score of zero; indeed, this profile is a basis vector for the effective space.  (Compare to if we had used $w_1$ for all three orbits, where we would\footnote{As with the Borda subspace in \cite{CrismanPermuta,OrrisonSymmetry,SaariStruct1,SaariStruct2}.} get $W$ as the winner with others tied, as in Example \ref{ex:oneweight}.)

As strange as it may seem on first glance, this effective space makes perfect sense from the wreath product (committee) standpoint; if we have favored $W$, then the only committee disjoint with it is $Z$, which should thus be disfavored.  The same thing would thus happen if we had made the same profile, but with $X$ as the favored committee (where $Y$ would have been the unique loser).
\end{example}

\begin{example}\label{ex:wine3}
Suppose ${\bf w}_1=[1/2,-1/2,-1/2,1/2]^T$, ${\bf w}_2=[1/2,-1/2,1/2,-1/2]^T$, and ${\bf w}_3=[1/2,1/2,-1/2,-1/2]^T$; this is Example \ref{ex:oneweight} but with weights moved as in Remark \ref{rmk:2wr2identdecomp}.  

A typical basis vector in the effective space is as follows: In the orbit order from Figure \ref{figure:2wr2pockets}, the first orbit has $1$ voter for each ranking with $W$ (or $Z$) in first or last place and $-1$ voter for the other rankings, and in the second and third orbits the same except $W$(/$Z$) in first or third place, and $W$(/$Z$) in first or second place, respectively.  %One can make analogous statements for this vector regarding $X,Y$.

Now instead of emphasizing a particular committee, the effective space emphasizes a particular \emph{structure}!  And since the procedure simply projects from this subspace and kills anything orthogonal to it, the result will \emph{necessarily} be either a tie $W/Z$ winning, $X/Y$ losing, or vice versa.  This is quite different from Example \ref{ex:oneweight}, and in an interesting way.
\end{example}

The procedure in the previous example may not be a very useful one, but it's possible without further constraints -- and \emph{any} system with weights not orthogonal to these has some of its behavior.  Notice the importance a decomposition makes!  

%\begin{example}
%The same example, but with $[1,-1/3,-1/3,-1/3]^T$ as the weights for the first orbit, behaves the same whether or not we have one ${\bf w}$ or one for each orbit.  This follows directly from Remark \ref{rmk:toomany}.  Its decomposition is isomorphic to the so-called Borda or Basic subspace\footnote{Indeed, it will be a linear combination of several such, including the spaces known as Reversal/Sym.} from `ordinary' voting theory, and (again, in both cases) is just plurality with a sum-zero weighting vector.
%\end{example}

Most interesting would be systems which had \emph{unrelated} weights for the different orbits.

\begin{example}
Suppose that one wanted to create a system that `rewarded' ballots which made the `logical' choice of putting disjoint committees in first and last place.  One could simply use Borda weights $w_1$ for the first orbit, and then use \emph{zero} for all other weights!  

Needless to say, the second two orbits would have zero effective space; only people who voted `correctly' would even have usable ballots.
\end{example}

\begin{example}
Suppose our motivation is as in the previous example, but we did not want to completely ignore `incorrect' ballots.  A plausible set of weights would be ${\bf w}_1=[1,1/2,-1/2, -1]^T$, but ${\bf w}_2={\bf w}_3=[1,-1/3,-1/3,-1/3]^T$.  This, roughly speaking, is a Borda count on the first orbit but plurality on the other two orbits.

Applying Proposition \ref{prop:2wr2forpockets} to this case, we see that the kernel in the first orbit, as desired, will be $S^{((2),\emptyset)}\oplus S^{(\emptyset,(2))}$.  However, from the other two orbits the kernel will be (identically) $S^{((2),\emptyset)}$.  So in some sense, even though we might be trying to \emph{minimize} the effect of the `undesired' type of ballots, in terms of the dimension of the effective spaces, they have \emph{more} impact in the sense of which profiles are taken account of.

\end{example}

\begin{example}
In the previous example, there really isn't any particular reason to pick the scaling we did, where $1$ is the highest number of points a ranking can get from a given ballot.   Suppose the motivation for these weights was instead wanting to imitate combining Borda and plurality as nearly as possible.  

In that case, one could point out that the $-1$ given to the least preferred ranking in the first orbit `punishes' it more than the $-1/3$ in the other two orbits.  In an effort to rectify this potential issue, we could normalize differently and use ${\bf w}_1=[1,1/2,-1/2, -1]^T$, but ${\bf w}_2={\bf w}_3=[1,-1,-1,-1]^T$.  But now the second weighting vector is no longer sum-zero!  From the point of view of the \emph{differentials} in points assigned to $W,X,Y,Z$ it behaves the same as the weighting vector ${\bf w}_2={\bf w}_3=[1.5,-0.5,-0.5,-0.5]^T$, which keeps the same problem in addition to adding the problem of giving more points to the most-preferred option!  This is probably the opposite effect\footnote{For this vector $x_2\neq 0$, so the effective space is four-dimensional for ${\bf w}_{2,3}$ too.} of that desired.
\end{example}

The point should be clear; there are a \emph{very} large number of possible voting rules.  On the one hand, it allows for a lot of flexibility in designing weights to `reward' certain behavior, but on the other hand moving even a little away from a `normal' voting system could yield many violations of standard voting axioms like Pareto.

\section{Conclusion}\label{sec:conclusion}

\subsection{Future Work}
Needless to say, there is a lot of room for further research.  We have scrupulously avoided even hinting at what axioms (other than neutrality and anonymity) might be relevant; as seen in the last few examples, even Pareto may need to be reconsidered when multiple orbits are involved.  

For instance, it does not seem obvious at all what we should do with ballots that do not put disjoint committees in opposite ranks.  Nevertheless, \cite{RatliffPublicChoice} and \cite{AliceBob} make it clear there are serious separability issues in real-life preferences, so such preferences \emph{will} appear.  Our model of examining the rules in this way gives the widest possible leeway for future analysis, and we make no apology for it.

Another area for future work is to use the strategy of symmetry (and/or representations) more explicitly in other combinatorial contexts.  Two places come immediately to mind.
\begin{itemize}
\item One can use a similar framework with the direct product to consider `diverse' committees directly, as in \cite{RatliffSaariComplexities} and \cite{RatliffSelectingDiverse}.  Here, we would \emph{not} consider all wreath symmetries, as that would destroy the diversity criterion.  Rather, the direct product may be the appropriate group.
\item The work in \cite{ElkindEtAlMultiwinner}, \cite{SkowronEtAlAxiomatic}, and \cite{FaliszewskiEtAlNewChallenge} is further from the model presented here, but clearly the committee scoring functions are analogous to the ones here.  The space of profiles would now be, as in normal voting, the regular representation of $S_n$, whereas the output might differ.  For instance, the $k$-voting rules would presumably use a space isomorphic to that decomposed so effectively for cooperative games in Kleinberg and Weiss and elsewhere\footnote{See \cite{KWAlgGames}, and the survey \cite{CrismanOrrison} for more references, including to some quite recent authors.}.
\end{itemize}

We look forward to seeing much more activity in the area of more unusual combinatorial settings such as these, especially connecting the algebraic framework with a more traditional axiomatic one.

\subsection{Funding Sources and Acknowledgements}

This work was supported by the Gordon College Faculty Development Committee and Provost's Office via a sabbatical leave.  The SQuaRE program at the American Institute of Mathematics supported this work indirectly by providing a good working space for improving my overall knowledge about representations and voting, though not this specific research.

Finally, thanks are due to the authors of \cite{LeeThesis} and \cite{BarceloEtAl} for many clear examples, to my SQuaRE team for moral support, as well as to the entire development team of Sage Math (\cite{Sage}), software which made discovering and verifying many of the results here much, much easier.   

\appendix

\section{Proofs}\label{sec:proofs}

\subsection{Proofs from Section \ref{sec:results}}\label{subsec:resultsproofs}

\begin{proof}[Proof of Theorem \ref{Correct3.8}]\label{proof:Correct3.8}
Recall (e.g.~from \cite{Sagan}) that the \emph{natural representation} of $S_m$ on the set of integers $[m]=\{1,2,\ldots ,m\}$ decomposes canonically as $\mathcal{N}\simeq S^{(m)}\bigoplus S^{(m-1,1)}$, where $S^\lambda$ is the irreducible $\mathbb{Q}S_m$-module indexed by the partition $\lambda \vdash m$.  

Character computations in \cite{BarceloEtAl} show $R_{m,n}$ is related very closely to $\otimes^n \mathcal{N}$.  However, \cite{JamesKerber} 4.3.9 is misstated as applying to a \emph{induced} representation rather than an \emph{extended} representation.  Extension changes the dimension, but here we will simply consider the $\mathbb{Q}S_m^n$-module $\otimes^n \mathcal{N}$ \emph{as a $\mathbb{Q}S_m\wr S_n$-module}, which follows because one can simply permute the coordinates of $\otimes^n \mathcal{N}$ for the action of $\pi$.  Hence\footnote{\cite{JamesKerber} 4.3.9 says in general that for a representation $D$ of $S_m$, $\chi_{\left(\otimes^n D\right)^\sim}=\prod_{\nu=1}^{c(\pi)}\chi_D(g_\nu(\sigma;\pi))$.}  we can say $$R_{m,n}\simeq \left(\otimes^n \mathcal{N}\right)^\sim$$ where the tilde denotes considering the $\mathbb{Q}S_m^n$-module to be over the larger group.

The remainder can follow \cite{BarceloEtAl} quite closely.   We can distribute the tensor product over this sum and see that $$R\simeq \left( \bigoplus_{k=0}^n\bigoplus_{\binom{n}{k}} \left(  S^{(m)^{\otimes^{n-k}}}\otimes S^{(m-1,1)^{\otimes^k}} \right)\right)^\sim\, ,$$ where each $S^{(m)^{\otimes^{n-k}}}\otimes S^{(m-1,1)^{\otimes^k}}$ is isomorphic \emph{but not equal} to each other under the natural permutations under $S_n$ of subsets of $[n]=\{1,2,\ldots ,n\}$ of size $k$ corresponding to the $S^{(m-1,1)}$ modules.  The action of $S_n$ clearly restricts to the sums for each $k$  separately, so we may write $$\bigoplus_{k=0}^n\left( \bigoplus_{\binom{n}{k}} \left(  S^{(m)^{\otimes^{n-k}}}\otimes S^{(m-1,1)^{\otimes^k}} \right)\right)^\sim\, .$$

Now consider for a moment any \emph{one} of the $S^{(m)^{\otimes^{n-k}}}\otimes S^{(m-1,1)^{\otimes^k}}$; a generic one may look like $S^{(m)}\otimes S^{(m-1,1)}\otimes S^{(m)}\otimes \cdots \otimes S^{(m)}\otimes S^{(m-1,1)}$.  This has a Young/inertia subgroup\footnote{See \cite{JamesKerber} for more details.} $Y_k\simeq S_{n-k}\times S_k$.  Notice that there is still a natural action of $Y_k$ on $S_m^n$, so that the wreath product $S_m\wr Y_k$ makes sense.  Its order is $m^n k!(n-k)!$, which means that $[S_m\wr S_n:S_m\wr Y_k]=\binom{n}{k}$, which would be the number of copies of the smaller module needed to create an induced module over $S_m\wr S_n$.

Returning now to $\left(\bigoplus_{\binom{n}{k}} \left(  S^{(m)^{\otimes^{n-k}}}\otimes S^{(m-1,1)^{\otimes^k}} \right)\right)^\sim$, let us consider the way in which this is \emph{extended} to a $S_m\wr S_n$-module.  If we had taken \emph{one} of the $S^{(m)^{\otimes^{n-k}}}\otimes S^{(m-1,1)^{\otimes^k}}$ as an $S_m^n$-module, then \emph{extended} that in the obvious way to a $S_m\wr Y_k$-module, and then \emph{induced} it to $S_m\wr S_n$, we would get precisely the same thing (because $S^{(m)^{\otimes^{n-k}}}\otimes S^{(m-1,1)^{\otimes^k}}$ stands in notationally for all $\binom{n}{k}$ possible orderings which inducing would need).  That is, we may write $$\left(\bigoplus_{\binom{n}{k}} \left(  S^{(m)^{\otimes^{n-k}}}\otimes S^{(m-1,1)^{\otimes^k}} \right)\right)^\sim=\left(  \left( S^{(m)^{\otimes^{n-k}}}\otimes S^{(m-1,1)^{\otimes^k}} \right)^\sim  \right)\ind^{S_m\wr S_n}_{S_m\wr Y_k}$$

Finally, as in (3.10) of \cite{BarceloEtAl} we may tensor with a trivial representation.  To be precise, the trivial representation $S^{(n-k)}\otimes S^{(k)}$ of $Y_k$ may be extended to $S_m\wr Y_k$ by acting on the individual $S_m$ components trivially, but permuting them via any $\pi\in Y_k$.  So we can write this same representation $$\left(  \left( S^{(m)^{\otimes^{n-k}}}\otimes S^{(m-1,1)^{\otimes^k}} \right)   \otimes   \left(S^{(n-k)}\otimes S^{(k)}\right)' \right)\ind^{S_m\wr S_n}_{S_m\wr Y_k}$$ and the rest of the proof from \cite{BarceloEtAl} goes through.
\end{proof}

\begin{proof}[Proof of Proposition \ref{prop:committeek}]
Let $C$ and $C'$ be as in the statement, and let $C\cap C'=D'$ of cardinality $f'$.  For each subset $c_k$ compatible with $C$, let $c'_k=C'\cap c_k$, of cardinality $f'_k$.  

Now consider another committee $C''$ such that $C\cap C''=D''$ and $f''=f'$ as well.  I claim there is an element $(\sigma;\pi)\in S_m\wr S_n$ such that $(\sigma;\pi)(C')=C''$ so that the multiset of cardinalities of the form $f''_k$ of the $c''_k=C''\cap c_k$ is the same (as multisets) as the collection of $f'_k$.  

First, pick any $\pi$ such that $\pi$ restricted to $D'$ yields $D''$.  Since the action of $S_n$ on the set of departments is clearly transitive and the cardinalities of these subsets is the same, this is possible.  Next, for each $i\in [n]$ suppose that $\pi(i)=j$, $C'$ includes $i_{k_{i'}}$ and $C''$ includes $j_{k_{j''}}$.  Then let $\sigma_i(k_{i'})=k_{j''}$ as well as $\sigma_i(k_i)=k_j$ (and otherwise arbitrary).  These restrictions are compatible on $D'$ (so that $(\sigma;\pi)(D')=D''$) and, just as importantly, preserve the intersections, so that $C\cap (\sigma;\pi)(C')=D''$.

From this construction, it is immediate that $(\sigma;\pi)(c'_k)$ is a set of committee members such that the size of the intersection with $C$ is preserved (though not necessarily the specific elements); in fact, it would be $(\sigma;\pi)(c_k)\cap (\sigma;\pi)(C'')$, which is one of the $c''_k$ sets.

Now we return to $b_C$.  This is defined to be the sum of $b_{c_k}$ of the types in Proposition \ref{prop:kronecker}, for precisely these $c_k$.  By this same proposition, its numerical value at the entry for any $C'$ must be the sum of numbers of the form $(m-1)^{f'_k}(-1)^{k-f'_k}$, for all $c_k$.  However, we have just shown that the multiset of numbers $f'_k$ themselves is invariant, given the cardinality of the intersection $C\cap C'$.
\end{proof}

\begin{proof}[Proof of Proposition \ref{prop:level2component}]
Certainly the constant term corresponds to choosing all $\binom{n}{2}$ subcommittees of size $k=2$, and checking whether they fully agree with the reference committee, in which case they get $(m-1)^{k-0}(-1)^0=(m-1)^2$ points from each basis vector summed in Proposition \ref{prop:committeek} from Proposition \ref{prop:kronecker}.

Suppose we disagree in precisely $d=1$ department.  Then actually $n-1$ of the subcommittees of size two must disagree, but only in one spot, so they received $(m-1)^{k-1}(-1)^1=-(m-1)$ points.  This is a net loss of $m(m-1)$ points for each of these $n-1$ possibilities.

Beyond this, each time we disagree with $C$ in one more department ($d$ increases by one) then $n-d$ more of the subcommittees of size two disagree in one department.  That means the \emph{change} in the change is by $m(m-1)$ \emph{less}.  On the other hand, each time $d$ increases (starting with $d=2$) we add one more subcommittee of size two which disagrees \emph{completely} with the reference committee, all of which receive $1$ point (so, a change from $-(m-1)$ to $+1$).  This means the change in the change overall is $m(m-1)+m=m^2$ with each increase in $d$.

Now using Newton's interpolation formula % get a standard reference if they ask
for finite differences\footnote{A very useful analogue of antidifferentiation for discrete calculus.}, and noting that $2!=2$ and the falling factorial $(d)_2=d(d-1)$, we get our result.
\end{proof}

\subsection{Proofs from Section \ref{sec:profile}}\label{subsec:proofsprofile}

\begin{proof}[Proof of Theorem \ref{thm:expandedBarceloEtAl4.2}]
Most of the proof goes through.  First, in the construction of $L$, rather than saying ${\bf w}\in \oplus S^{\lambda}$, it should really be that ${\bf \tilde{w}}\in \oplus S^{\lambda}$ such that ${\bf \tilde{w}}$ restricted to $R$ is isomorphic to ${\bf w}$, and likewise for ${\bf w}_i$.  

More importantly, in the proof of Theorem 4.2 of \cite{BarceloEtAl}, Jacobson density is used to create a profile $\p\in \mathbb{Q}S_m\wr S_n$, but then is discussed as if it were in $\mathbb{Q}S_{m^n}$; however, the fact that $\p$ is actually in that subalgebra is the novel part of their theorem.   Further, in the decomposition into orbits of $P_{m,n}$, one can act by some other element of $S^{m^n}$ to obtain a new profile that otherwise will have the same properties.
\end{proof}

\subsection{Proofs from Section \ref{sec:systems}}\label{subsec:proofssystems}

\begin{proof}[Proof of Proposition \ref{prop:bordaweightsborda}]
We use the basis from Fact \ref{fact:otherbordabasis}.  Let us project a sum-zero version of the Borda weighting vector $$\left[\frac{m^n-1}{2},\frac{m^n-3}{2},\ldots ,\frac{1-m^n}{2}\right]^T$$ to each of the $B_d$ subspaces; if the sum of those projections is the weighting vector, since the $B_d$ are orthogonal, we are done.

Consider vectors of the form $\frac{m-2i+1}{2}m^{n-d}$ as the entry for every entry corresponding to a committee of the form $(\ldots , d_{i},\ldots)$, for $1\leq i\leq m$.  Since $\sum_{c=1}^{m-1}-\frac{m-2c+1}{2}=\frac{1-m}{2}$, this vector is certainly a sum of elements of the alternate basis.  If we add all these vectors, the entry corresponding to $C'=(1_{j_1},2_{j_2},\ldots,n_{j_n})$ in that sum is $$\sum_{d=1}^n \frac{m-2j_d+1}{2}m^{n-d}\; .$$  By reindexing this we obtain $$\sum_{e=0}^{n-1} \frac{m-2j_{n-e}+1}{2}m^{e}$$ which may be thought of as a `balanced $n$-ary expansion' in the lexicographic order of a unique (half-)integer from $\frac{1-m^n}{2}$ to $\frac{m^n-1}{2}$ -- precisely the ones in our sum-zero Borda weighting vector.
\end{proof}

\bibliographystyle{amsplain}
\bibliography{refs}

\end{document}